\documentclass[12pt]{amsart}
\usepackage{amsmath,amssymb,amsthm,bbm,color,setspace,textcomp,layout,graphicx,fullpage}
\usepackage[foot]{amsaddr}

\theoremstyle{plain}
\newtheorem{thm}{Theorem}[section]
\newtheorem{prop}[thm]{Proposition}
\newtheorem{cor}[thm]{Corollary}
\newtheorem{lemma}[thm]{Lemma}
\theoremstyle{definition}
\newtheorem{defn}[thm]{Definition}
\theoremstyle{remark}
\newtheorem*{rem}{Remark}
\theoremstyle{definition}
\newtheorem{assume}[thm]{Assumption}

\renewcommand{\Re}{\text{Re}}

\title[Branching Processes and the Luria-Delbr\"uck Experiment]{Age-dependent Branching Processes and Applications to the Luria-Delbr\"uck Experiment}

\author{Stephen Montgomery-Smith}
\address{Department of Mathematics, University of Missouri, Columbia MO 65211.}
\email{stephen@missouri.edu}
\author{Hesam Oveys}
\address{Courant Institute of Mathematical Sciences, New York University, New York, NY 10012.}
\email{ho15@nyu.edu}



\begin{document}

\begin{abstract}
Microbial populations adapt to their environment by acquiring advantageous mutations, but in the early twentieth century, questions about how these organisms acquire mutations arose. The experiment of Salvador Luria and Max Delbr\"uck that won them a Nobel Prize in 1969 confirmed that mutations don't occur out of necessity, but instead can occur many generations before there is a selective advantage, and thus organisms follow Darwinian evolution instead of Lamarckian. Since then, new areas of research involving microbial evolution have spawned as a result of their experiment. Determining the mutation rate of a cell is one such area. Probability distributions that determine the number of mutants in a large population have been derived by D. E. Lea, C. A. Coulson, and J. B. S. Haldane. However, not much work has been done when time of cell division is dependent on the cell age, and even less so when cell division is asymmetric, which is the case in most microbial populations. Using probability generating function methods, we rigorously construct a probability distribution for the cell population size given a life-span distribution for both mother and daughter cells, and then determine its asymptotic growth rate. We use this to construct a probability distribution for the number of mutants in a large cell population, which can be used with likelihood methods to estimate the cell mutation rate.
\end{abstract}

\maketitle

\section{Introduction}

\subsection{The Luria-Delbr\"uck Experiment}

In the early twentieth century, questions about how microorganisms acquire advantageous mutations arose. In 1943, biologists Salvador Luria and Max Delbr\"uck conducted an experiment in order to determine whether mutations occurred out of necessity or could occur many generations before there was a selective advantage. This experiment, dubbed the ``Luria-Delbr\"uck Experiment," helped them win a Nobel Prize in 1969 (see \cite{luriaDelbruck} and \cite{g-smith-et-al}).

In their experiment, Luria and Delbr\"uck grew bacteria in a non-selective medium in multiple tubes for a period of time until they all reached a certain cell density. Then, they plated the cells from each tube on different plates of a selective-medium containing a bacterial virus. Cells that showed resistance to the virus had acquired a virus-resistant mutation. If cells evolved according to a post-exposure hypotheses such as Lamarckian evolution, where cells acquire mutations in response to their environment, then the number of mutants in each plate would follow a Poisson process where the mean is equal to the variance, making plates with a large number of mutants highly unlikely. But in the experiment of Luria and Delbr\"uck, there were ``jackpots," meaning there were plates with an unusually large number of surviving cells. The only conclusion they could make was that mutation occurred before the cells were plated on the selective medium containing the virus, and thus Charles Darwin's theory of natural selection applied to microorganisms.

Since then, new areas of research involving microbial evolution have spawned from the Luria-Delbr\"uck Experiment which are still being studied today. Since cells can acquire mutations before there is a selective advantage in their environment, questions about their mutation rate have risen. However, traditional methods involving significance tests can't be used due to the high variability of the data. Instead, we can use likelihood methods, but in order to do so, we need a probability distribution for the mutant cell population size as a function of the unknown mutation rate.

Though unpublished originally, a probability distribution for the number of mutants was presented by John B. S. Haldane (see \cite{sarkarHaldane}). However, Haldane's model had two major issues: it assumed all cells divide synchronously, and the distribution was computationally inefficient. In 1949, D. E. Lea and C. A. Coulson constructed a generating function for the number of mutants in a large population that had a closed-form solution, so computing the probability coefficients was much more efficient. However, they assumed that all cells grow symmetrically with a life-span distribution that was exponential, which is a very broad assumption for how cells grow (see \cite{leaCoulson}).

In this paper, we look to extend on the ideas of Haldane, Lea, and Coulson and develop a probability distribution for the mutant cell population size where we have control of a cell's life-span distribution. In addition, we develop a distribution for asymmetric cell division, where a cell divides into a mother and daughter cell with different life-span distributions.

\subsection{The problem}

Suppose you start with a single cell, and this cell undergoes binary division and divides into two cells. As time passes, the total cell population will grow, but depending on when each cell divides, the population will vary. If a cell divides into two cells identical to itself (symmetrical cell division), then the two children cells will divide similarly to its parent. On the other hand, how will the cell population grow if a cell divides into two cells, where one is identical to itself, but the other is not (asymmetrical cell division)? In addition, suppose during any cell division, a mutation can occur. If this mutation is passed through all the children spawned from the mutated cell, what can we say about the distribution of the mutant cell population size in a sufficiently large cell population?

The solution to this problem is given in Theorems~\ref{multiMutantDist} and~\ref{singleMutantDist}, in the form of probability generating functions.  In Section~\ref{MutantExamples} we show how to explicitly compute these probabilities.

\section{Probability generating functions} \label{pgf}

A very brief overview of probability generating functions is given to ensure the reader is familiar with the basic properties. The first thing to note is that a probability generating function can only be constructed for discrete random variables outputting non-negative integers.

\begin{defn}
Let $X$ be a discrete random variable outputting non-negative integer values. The \emph{probability generating function} (or simply \emph{generating function}) of $X$ is the function
\begin{equation}
G_X(z) := \mathbb{E}\left[z^X\right] = \sum_{k=0}^\infty \Pr(X=k) z^k.
\end{equation}
\end{defn}

Probability generating functions are power series with non-negative coefficients such that their sum is 1, so their radius of convergence is always at least 1.

\begin{prop}
If $X$ is a discrete random variable outputting non-negative integer values with generating function $G_X(z)$, then
\begin{equation}
\mathbb{E}[X] = G_X'(1).
\end{equation}
\end{prop}

Below, we define joint probability generating functions, since we will use them in Section~\ref{multiTypeGrowth}.

\begin{defn}
Let $X$ and $Y$ be jointly distributed discrete random variables outputting non-negative integer values. Then the \emph{joint probability generating function} is
\begin{equation}
G_{X,Y}(x,y) := \mathbb{E}\left[ x^X y^Y \right] = \sum_{k=0}^\infty \sum_{j=0}^\infty \Pr(X=k, \, Y=j) x^k y^j.
\end{equation}
\end{defn}

\begin{prop}
If $X$ and $Y$ are independent random variables, both outputting non-negative integer values, then the generating function of $X+Y$ is
\begin{equation}
G_{X+Y}(z) = G_X(z) G_Y(z).
\end{equation}
\end{prop}

The following two propositions will be used extensively in Section \ref{mutants} when constructing the generating function for the mutant cell population.

\begin{prop}
If $X=X_k$ are independent and identically distributed random variables outputting non-negative integer values, and
\begin{equation}
Z = \sum_{k=0}^N X_k,
\end{equation}
where $N$ is an independent random variable outputting non-negative integer values, then the generating function of $Z$ is
\begin{equation}
G_Z(z) = G_N(G_X(z)).
\end{equation}
\end{prop}

\begin{prop}
If $X$ is a Poisson random variable with parameter $\lambda$, then the generating function for $X$ is
\begin{equation}
G_X(z) = e^{\lambda(z-1)}.
\end{equation}
\end{prop}

For a more in-depth look at generating functions, one can refer to almost any book on probability theory, though we only use the above definitions and properties in this paper.

\section{Cell growth under asymmetric division} \label{multiTypeGrowth}

It is of great interest to biologists to model the growth of a cell under asymmetric division; that is, when it divides into two types: a mother cell and a daughter cell. The mother cell is a copy of the parent cell, while the daughter cell typically takes time to grow into a mother cell. Therefore, the life-span distribution for mother cells is different, and often their life-span is shorter than daughter cells.

In this section, we will construct time-dependent generating functions \eqref{multiIntEqns} for the mother and daughter cell population sizes when we start with exactly one mother cell and when we start with exactly one daughter cell. These generating functions will be in the solution of an integral equation, so we will proceed to show existence and uniqueness of solutions in Section~\ref{multiExistence}. We conclude the section with some examples of generating functions using different life-span distributions.

In \cite{Oveys-thesis}, the symmetric case is expounded in full detail before the asymmetric case is described.  In this paper, in the interests of brevity, we only explain the asymmetric case, since the symmetric case follows as a special case.

\subsection{Preliminaries and assumptions}

We will start by stating our assumptions about asymmetric cell division.

\begin{assume}Cells have the following properties:
\begin{enumerate}
\item there are exactly two types of cells: mother cells and daughter cells;
\item all cells are independent of each other, mother cells are identical to other mother cells, and daughter cells are identical to other daughter cells;
\item cell life-span for mother cells and daughter cells are strictly positive, real-valued random variables $\mathcal{T}_x$ and $\mathcal{T}_y$, respectively, with distributions
\begin{equation}
\begin{aligned}
P(t) = \Pr(\mathcal{T}_x \leq t) \\
Q(t) = \Pr(\mathcal{T}_y \leq t)
\end{aligned}
\end{equation}
respectively;
\item at the end of a cell's life, both mother cells and daughter cells will divide into one mother cell and one daughter cell.
\end{enumerate}
\end{assume}

Let $X_t$ and $Y_t$ be random variables representing the mother and daughter cell populations at time $t \geq 0$, respectively. When we start with exactly one mother cell and no daughter cells, which is the main case of interest, we will have $X_0 = 1$ and $Y_0 = 0$. On the other hand, when we start with exactly one daughter cell and no mother cells, we will have $X_0 = 0$ and $Y_0 = 1$. Note that $X_t$ and $Y_t$ are \emph{not} independent processes.

\subsection{Constructing the generating function}

In order to derive the generating function for the cell population at any time $t \geq 0$, we will first construct a model in discrete-time, then divide our time increments infinitesimally small, and finally take limits to derive a continuous-time model.

Suppose first that time is discrete and the life-span random variables for mother and daughter cells, $\mathcal{T}_x$ and $\mathcal{T}_y$, only output positive integer values. For clarity, when we say a cell divides at time $t$, we mean the cell population has increased by one at time $t+1$.

Define
\begin{equation}
\begin{aligned}
f(t,x,y) &:= \mathbb{E} \left[ x^{X_t} y^{Y_t} \, | \, X_0 = 1, \, Y_0 = 0 \right] \\
g(t,x,y) &:= \mathbb{E} \left[ x^{X_t} y^{Y_t} \, | \, X_0 = 0, \, Y_0 = 1 \right]
\end{aligned}
\end{equation}
to be the joint generating functions of $X_t$ and $Y_t$, with different initial values.  Then it can easily be shown that
\begin{equation}
\begin{aligned}
f(t,x,y) &= x\Pr(\mathcal{T}_x > t) + \sum_{k=1}^t f(t-(k+1),x,y)g(t-(k+1),x,y) \Pr(\mathcal{T}_x = k) \\
g(t,x,y) &= y\Pr(\mathcal{T}_y > t) + \sum_{j=1}^t f(t-(j+1),x,y)g(t-(j+1),x,y) \Pr(\mathcal{T}_y = j).
\end{aligned}
\end{equation}

Now, to get the formulas in continuous time, we take a limit of the discrete time formula as the discrete-time increments converge to zero, and we obtain
\begin{equation}
\begin{aligned}
\displaystyle f(t,x,y) = x(1-P(t)) + \int_0^t f(t-\tau,x,y)g(t-\tau,x,y)\,dP(\tau) \\
\displaystyle g(t,x,y) = y(1-Q(t)) + \int_0^t f(t-\tau,x,y)g(t-\tau,x,y)\,dQ(\tau).
\end{aligned} \label{multiIntEqns}
\end{equation}
Since the integrals are Lebesgue-Steiltjes integrals, it is important to emphasize that $\int_0^t$ denotes integration over the closed interval $[0,t]$.

\subsection{Existence and uniqueness} \label{multiExistence}

A variant of arguments given in \cite{harris}, Chapter~VI, \S9 will also work, but we present a slightly different approach.

Showing that there exists two unique generating functions $f$ and $g$ that satisfy \eqref{multiIntEqns} given any life-span distributions $P$ for mother cells and $Q$ for daughter cells will follow by constructing suitable functions spaces and a map so we can use the Banach Fixed-Point Theorem.  In the proceeding sections, we will always assume $0 < r < 1$. 

Let $H^\infty(B(0,r)^2)$ represent the space of all holomorphic functions of two complex variables bounded on $\partial B(0,r)^2$ with the usual norm
\begin{equation}
{\| \cdot \|}_{H^\infty} := \sup_{(z_1,z_2) \in B(0,r)^2} |(\cdot)(z_1,z_2)|.
\end{equation}

\begin{defn}
\label{defn H}
Define the subset $\mathcal{H}_r$ of $H^\infty(B(0,r)^2)$ such that for each $g \in \mathcal{H}_r$,
\begin{enumerate}
\item $c_{k,j}(g) \geq 0$ for integers $k,j \geq 0$ and
\item \label{2nd cond} $\displaystyle \sum_{k=0}^\infty \sum_{j=0}^\infty c_{k,j}(g) \le 1$,
\end{enumerate}
where
\begin{equation}
\label{cauchy-int}
c_{k,j}(g) = \frac{1}{k!j!}\frac{\partial^{k+j}g}{\partial x^k \partial y^j}(0,0) = \frac{1}{(2\pi i)^2} \int_{z_1 \in C(0,r)} \int_{z_2 \in C(0,r)} \frac{g(z_1,z_2)}{z_1^{k+1}z_2^{j+1}}\,dz_2 \, dz_1,
\end{equation}
is the coefficient of $x^k y^j$ of $g$ in its power series expansion centered at 0, where $x$ and $y$ are the function parameters.  Here $C(0,r)$ represents the path in $\mathbb C$ along a circle centered at the origin of radius $r$ traversed once counter-clockwise.
\end{defn}

\begin{prop}
$\mathcal{H}_r$ is a separable, complete metric space with respect to the norm ${\|\cdot\|}_{H^\infty}$.
\end{prop}

\begin{proof}
To show completeness, we need prove is that $\mathcal{H}_r$ is closed in $H^\infty(B(0,r)^2)$.    By equation~\eqref{cauchy-int} it follows that the functions $c_{j,k}$ are continuous on $H^\infty$, and hence the only difficulty is to show Definition~\ref{defn H}, Part~\eqref{2nd cond} is preserved by convergent sequences in $H^\infty$.  But this follows by noting that it is equivalent to
\begin{equation}
\sum_{k,j=0}^M c_{k,j}(g) \leq 1 \quad\text{for all $M \in \mathbb N$}
\end{equation}
Separability follows from part~\eqref{2nd cond} since it follows that the set of polynomials in $\mathcal{H}_r$ is dense in $\mathcal{H}_r$.
\end{proof}

\begin{defn}
Define $\mathcal{P}_{[0,\infty)}$ to be the set of all partitions of $[0,\infty)$. We say a function $f: [0,\infty) \to H^\infty(B(0,r)^2)$ has \emph{bounded variation} on $[0,\infty)$ if
\begin{equation}
V(f) :=\sup_{S \in \mathcal{P}_[0,\infty)}\sum_{k=1}^M {\|f(t_k)-f(t_{k-1})\|}_{H^\infty}
\end{equation}
is finite, where $S$ is a partition $0 \le t_0 \leq t_1 \leq \dots \leq t_M < \infty$.
\end{defn}

We will denote the the essential supremum of a function $f:[0,\infty) \to H^\infty(B(0,r)^2)$ using the usual norm
\begin{equation}
{\| f \|}_{L^\infty(H^\infty)} := \operatorname*{esssup}_{t\in [0,\infty)} {\|f(t)\|}_{H^\infty}.
\end{equation}

\begin{defn} \label{BV}
Define $BV([0,\infty),H^\infty(B(0,r)^2))$ to be the space of all functions mapping from $[0,\infty)$ to $H^\infty(B(0,r)^2)$ with bounded variation on $[0,\infty)$ with norm
\begin{equation}
{\| \cdot \|}_{BV(H^\infty)} := {\| \cdot \|}_{L^\infty(H^\infty)} + V(\cdot). \label{Brnorm}
\end{equation}
\end{defn}

The following result can be proved using standard techniques (see \cite{Oveys-thesis}).

\begin{prop}
$BV([0,\infty),H^\infty(B(0,r)^2))$ is a separable Banach space, consisting of bounded Borel measurable functions from $[0,\infty)$ to $H^\infty(B(0,r)^2)$.
\end{prop}

Thus for any distribution $P$, we can define integration for measurable functions $f:[0,\infty) \to H^\infty(B(0,r)^2)$ using the Bochner integral on $L^1(([0,\infty),dP),H^\infty(B(0,r)^2))$ \cite{diestel}.

\begin{defn}
Define $\mathcal{B}_{r,m}$ to be the subspace of $BV([0,\infty),H^\infty(B(0,r)^2))$ such that for all $f \in \mathcal{B}_{r,m}$,
\begin{enumerate}
\item $f(t) \in \mathcal{H}_r$ for all $t \in [0,\infty)$,
\item $c_{0,0}(f(t)) = 0$ for all $t \in [0,\infty)$, and
\item $V(f) \leq m$.
\end{enumerate}
\end{defn}

\begin{prop}
$\mathcal{B}_{r,m}$ is a complete metric space with respect to the norm ${\| \cdot \|}_{BV(H^\infty)}$, and thus the product space $\left( \mathcal{B}_{r,m} \right)^2$ is a complete metric space with product norm 
\begin{equation}
{\left\| \left( (\cdot)_1, (\cdot)_2 \right) \right\|}_{{BV}^2(H^\infty)} := {\|(\cdot)_1\|}_{BV(H^\infty)} + {\|(\cdot)_2\|}_{BV(H^\infty)}.
\end{equation}
\end{prop}

\begin{defn}
Define
\begin{equation}
T:\left( \mathcal{B}_{r,m} \right)^2 \to \left( BV([0,\infty),H^\infty(B(0,r))) \right)^2
\end{equation}
to be a map such that for $(f,g) \in \left( \mathcal{B}_{r,m} \right)^2$,
\begin{align}
(T(f,g))(t,s) & := \left( (\cdot)_1 (1-P(t)) + \int_0^t f(t-\tau)g(t-\tau) \, dP(\tau), \right. \nonumber \\
& \quad \quad \left. (\cdot)_2 (1-Q(s)) + \int_0^s f(s-\tau)g(s-\tau) \, dQ(\tau) \right)
\end{align}
for all $(t,s) \in [0,\infty)^2$, where $(\cdot)_1$ and $(\cdot)_2$ represent the first and second parameters of functions in $H^\infty(B(0,r)^2)$, respectively.
\end{defn}

\begin{thm}[Banach Fixed-Point Theorem] \label{bfpt}
Suppose $X$ is a complete metric space with distance function $d(\cdot,\cdot)$ and $T:X \to X$ is a map such that there exists a constant $0 \leq \gamma<1$ where $d(T(x),T(y)) \leq \gamma d(x,y)$ for all $x,y \in X$. Then $T$ has a unique fixed-point.
\end{thm}

\begin{thm}
\label{fixed point}
If $\displaystyle 0 < m < \frac{1}{2}$ and $\displaystyle 0 < r < \min \left( \frac{1-2m}{4}, \frac{m}{2(1+m)} \right)$, then $T$ is a contraction mapping of $\left( \mathcal{B}_{r,m} \right)^2$ to itself, and hence $T$ has a unique fixed point.
\end{thm}

\begin{rem}
Since analytic functions are uniquely determined, as long as $r > 0$, we can extend our fixed point to converge on $B(0,1)^2$.
\end{rem}

\begin{thm} Let $(f,g)$ be the fixed point of $T$.  Then $f(t,1,1) = g(t,1,1) = 1$, and hence $f$ and $g$ are generating functions that solve equation~\eqref{multiIntEqns}.
\end{thm}

\begin{proof}
Clearly $f(t,1,1), g(t,1,1) \le 1$.  Next, since $P(0) = Q(0) = 0$, we see that
\begin{equation}
f(0,1,1) = g(0,1,1) = 1 .
\end{equation}
Let
\begin{equation}
t^* = \inf\{t \ge 0 : f(t,1,1) \ne 1 \text{ or } g(t,1,1) \ne 1 \}.
\end{equation}
Pick $0 < \epsilon < 1$.  There exists $\delta > 0$ such that
$P(\delta), Q(\delta) < \epsilon$.
Let
\begin{equation}
f_0 = \inf \{ f(t,1,1), g(t,1,1) : t \in [t^*,t^*+\delta) \} .
\end{equation}
From equation~\eqref{multiIntEqns}, we obtain that for $t \in [t^*,t^*+\delta)$
\begin{equation}
\begin{aligned}
f(t,1,1) &= 1-P(t) + \int_{t-t^* < \tau \le t} f(t-\tau,1,1) g(t-\tau,1,1) \, dP(\tau)
\\ & \phantom{{}\ge{}} + \int_0^{t-t^*} f(t-\tau,1,1) g(t-\tau,1,1) \, dP(\tau)
\\ &\ge 1 - P(t-t^*) + P(t-t^*) f_0^2
\end{aligned}
\end{equation}
and similarly
\begin{equation}
g(t,1,1) \ge 1 - Q(t-t^*) + Q(t-t^*) f_0^2.
\end{equation}
Hence
\begin{equation}
f_0 \ge 1 - \epsilon + \epsilon f_0^2 \quad \Rightarrow \quad (f_0-1)(\epsilon f_0 - 1 + \epsilon) \le 0
\end{equation}
from which it follows that $f_0 \ge 1$.
\end{proof}

\subsection{Series solutions to the generating function equation} \label{multiSeriesSoln}

If we are only concerned about the total cell population and not specifically the mother and daughter cell populations, we can write our generating functions $f(t,x,y)$ and $g(t,x,y)$ as simply $f(t,x) = f(t,x,x)$ and $g(t,x) = g(t,x,x)$. 

So, we can express our generating functions $f(t,x)$ and $g(t,x)$ as series
\begin{equation}
f(t,x) = \sum_{k=0}^\infty c_k(t)x^k, \qquad g(t,x) = \sum_{k=0}^\infty b_k(t)x^k,
\end{equation}
which are necessarily convergent for all $t \geq 0 $ when $x \in B(0,1)$. If $f(t,x)$ and $g(t,x)$ satisfy our integral equations \eqref{multiIntEqns}, then by matching up coefficients and noting $c_0(t) = b_0(t) = 0$ for all $t \geq 0$, we get
\begin{equation}
\begin{aligned}
c_k(t) &= \begin{cases}
\displaystyle 1-P(t) & k=1 \\
\displaystyle \int_0^t \sum_{j=1}^{k-1} c_{j}(t-\tau)b_{k-j}(t-\tau) \, dP(\tau) & k\geq 2
\end{cases} \\
\\
b_k(t) &= \begin{cases}
\displaystyle 1-Q(t) & k=1 \\
\displaystyle \int_0^t \sum_{j=1}^{k-1} c_{j}(t-\tau)b_{k-j}(t-\tau) \, dQ(\tau) & k\geq 2
\end{cases}.
\end{aligned}
\end{equation}

In Section~\ref{MutantExamples}, we will use these formulas to determine the distribution of the mutant cell population when $P$ and $Q$ are multi-phase distributions.

\section{Asymptotics of cell growth} \label{asymptotics}

\begin{defn} \label{lattice}
A probability distribution $P$ is a \emph{$\delta$-lattice distribution} if $P$ is constant except at jumps at multiples of some $\delta > 0$.
\end{defn}

Define the following integrals for $s \in \mathbb C$:
\begin{equation}
p^*(s) := \int_0^\infty e^{-st} \, dP(t), \quad
q^*(s) := \int_0^\infty e^{-st} \, dQ(t), \quad
\psi(s) := 1 - p^*(s) - q^*(s).
\end{equation}

Note
\begin{equation}
\psi'(s) = \int_0^\infty te^{-st} \, d(P(t)+Q(t)),
\end{equation}
is positive when $s\ge 0$ is real.

\begin{prop}\label{multiPsiRoot}
$\psi$ has a unique real root, $\alpha \in (0,\infty)$.
\end{prop}

\begin{proof}
This follows since $\psi(0) = -1$, $\lim_{s \to +\infty} \psi(s) = 1$, and $\psi'(s) > 0$ for $s \in [0,\infty)$.
\end{proof}

\begin{rem}
In the rest of this paper, unless otherwise stated, $\alpha$ will always refer to the unique real root of $\psi$.
\end{rem}

The goal of Section~\ref{asymptotics} is to prove the following result.

\begin{thm} \label{main-asymp}
If for any $\delta>0$ we have that either $P$ or $Q$ is not a $\delta$-lattice distribution, then there exists a non-negative random variable $V$ such that
\begin{equation}
X_t e^{-\alpha t} \to p^*(\alpha) V , \quad
Y_t e^{-\alpha t} \to q^*(\alpha) V
\end{equation}
where the convergence is in $L^2$.
\end{thm}

\subsection{The expectations $\mathbb{E}[X_t]$ and $\mathbb{E}[Y_t]$}

Define the following expectations:
\begin{equation}
\begin{aligned} 
m_f(t) &:= \mathbb{E}\left[X_t \, | \, X_0 = 1, \, Y_0=0\right] = f_x(t,1,1); \\
m_g(t) &:= \mathbb{E}\left[X_t \, | \, X_0 = 0, \, Y_0=1\right] = g_x(t,1,1); \\
n_f(t) &:= \mathbb{E}\left[Y_t \, | \, X_0 = 1, \, Y_0=0\right] = f_y(t,1,1); \\
n_g(t) &:= \mathbb{E}\left[Y_t \, | \, X_0=0, \, Y_0=1\right] = g_y(t,1,1).
\end{aligned}
\end{equation}

\begin{prop}
The expectations $m_f(t)$, $m_g(t)$, $n_f(t)$, and $n_g(t)$ satisfy the systems
\begin{equation}
\begin{aligned}
\displaystyle m_f(t) &= 1-P(t) + \int_0^t m_f(t-\tau) + m_g(t-\tau) \, dP(\tau) \\
\displaystyle m_g(t) &= \int_0^t m_f(t-\tau) + m_g(t-\tau) \, dQ(\tau) \\
\displaystyle n_f(t) &= \int_0^t n_f(t-\tau) + n_g(t-\tau) \, dP(\tau) \\
\displaystyle n_g(t) &= 1-Q(t) + \int_0^t n_f(t-\tau) + n_g(t-\tau) \, dQ(\tau).
\end{aligned} \label{mintEqns}
\end{equation}
\end{prop}
\begin{proof}
Suppose $x,y \in B(0,1)$. We can differentiate with respect to $x$ and $y$, both sides of both equations in \eqref{multiIntEqns}.  It is straightforward to bring the derivatives under the integral sign using the Cauchy integral formula and Fubini's Theorem.
 Since these generating functions and their derivatives have positive coefficients and are increasing in both the $x$ and $y$ parameters, we can let $x,y \to 1^-$.
\end{proof}

\begin{prop} \label{mnFinite}
The expectations $m_f(t)$, $m_g(t)$, $n_f(t)$, and $n_g(t)$ are non-decreasing and finite for $t \geq 0$.
\end{prop}
\begin{proof}
Since no cells can die, $m_f$, $m_g$, $n_f$, and $n_g$ must be non-decreasing. Since $P$ and $Q$ are life-span distributions with $P(0)=Q(0)=0$ and $P$ and $Q$ are right continuous, then for $\varepsilon > 0$, there exists a $\delta > 0$ such that $P(\delta) < \varepsilon$ and $Q(\delta) < \varepsilon$. Now, rewriting and bounding the system \eqref{mintEqns}, we get
\begin{equation}
\begin{aligned}
m_f(t) &\leq 1 + \varepsilon m_f(t) + m_f(t-\delta) + \varepsilon m_g(t) + m_g(t-\delta) \\
m_g(t) &\leq \varepsilon m_f(t) + m_f(t-\delta) + \varepsilon m_g(t) + m_g(t-\delta).
\end{aligned}
\end{equation}
So,
\begin{equation}
m_f(t)+m_g(t) \leq 1 + 2\varepsilon(m_f(t)+m_g(t))+2(m_f(t-\delta)+m_g(t-\delta)),
\end{equation}
and
\begin{equation}
m_f(t)+m_g(t) \leq \frac{1+2(m_f(t-\delta)+m_g(t-\delta))}{1-2\varepsilon},
\end{equation}
that is, if $m_f(t-\delta)$ and $m_g(t-\delta)$ are finite, then $m_f(t)$ and $m_g(t)$ are finite.  Since $m_f(0)+m_g(0) = 1$, we can inductively conclude $m_f(t)$ and $m_g(t)$ are finite for all $t \geq 0$, and the result is proven for $m_f$ and $m_g$. $n_f$ and $n_g$ are similar.
\end{proof}

\subsection{Convergence of $\mathbb{E}[X_t]e^{-\alpha t}$ and $\mathbb{E}[Y_t]e^{-\alpha t}$}

We will compute Laplace transforms, and then analyze their poles.  We need to first review some properties of the Laplace transform.

\begin{defn}
Let $h: [0,\infty) \to \mathbb{R}$. The \emph{Laplace transform} of $h$ is
\begin{equation}
(\mathcal{L}h)(s) := \int_0^\infty e^{-st} h(t) \, dt, \quad s \in \mathbb{C}.
\end{equation}
\end{defn}
\begin{rem}
We will often write $h^* := \mathcal{L}h$ to represent the Laplace transform of $h$.
\end{rem}

\begin{prop}[see \cite{widder}, page 92]
If $f,g \in L^1([0,R))$ for all $R>0$, then
\begin{equation}
\mathcal{L}(f \ast g) := (\mathcal{L}f)(\mathcal{L}g),
\end{equation}
provided all three transforms exist.
\end{prop}

Similarly we also have the following.

\begin{prop}
Let $\varphi:[0,\infty) \to \mathbb{R}$ be a function and $G$ be a probability distribution on $[0,\infty)$. Define
\begin{equation}
(\varphi \ast dG)(t) := \int_0^t \varphi(t-\tau) \, dG(\tau)
\end{equation}
and
\begin{equation}
g^*(s) := \int_0^\infty e^{-st} \,dG(t).
\end{equation}
Then,
\begin{equation}
\mathcal{L}(\varphi \ast dG) = (\mathcal{L}\varphi) g^*
\end{equation}
when all integrals converge.
\end{prop}

Let $s \in \mathbb{C}$ with $\Re(s) > 0$. Then,
\begin{equation}
\begin{aligned}
m_f^*(s) & := (\mathcal{L}m_f)(s) = \frac{1}{s} - \frac{1}{s}p^*(s) + m_f^*(s)p^*(s) + m_g^*(s)p^*(s), \\
m_g^*(s) & := (\mathcal{L}m_g)(s) = m_f^*(s)q^*(s) + m_g^*(s)q^*(s),
\end{aligned}
\end{equation}

Solving for $m_f^*(s)$ and $m_g^*(s)$, we get
\begin{equation}
\begin{aligned}
m_f^*(s) &= \frac{1-p^*(s)-q^*(s)+p^*(s)q^*(s)}{s(1-p^*(s)-q^*(s))} \\
m_g^*(s) &= \frac{(1-p^*(s))q^*(s)}{s(1-p^*(s)-q^*(s))}.
\end{aligned}
\end{equation}

Similarly,
\begin{equation}
\begin{aligned}
n_f^*(s) &= \frac{(1-q^*(s))p^*(s)}{s(1-p^*(s)-q^*(s))} \\
n_g^*(s) &= \frac{1-p^*(s)-q^*(s)+p^*(s)q^*(s)}{s(1-p^*(s)-q^*(s))}.
\end{aligned}
\end{equation}

If $P$ and $Q$ are lattice distributions, then it can be shown that there are infinitely many zeros of $\psi$ with $\Re(s) = \alpha$, but if $P$ or $Q$ is not a lattice distribution, then the following hold.

\begin{prop} \label{mstarAnalytic}
If $P$ and $Q$ are not both $\delta$-lattice distributions, and $\alpha + i \tau$ is a zero of $\psi$, then $\tau = 0$.  Hence $m_f^*$, $m_g^*$, $n_f^*$, and $n_g^*$ are analytic at $\alpha + i\tau$ when $\tau \neq 0$.
\end{prop}

\begin{proof}
Since $\alpha$ and $\alpha + i \tau$ are zeros, we have
\begin{equation}
\Re\left(\psi(\alpha) - \psi(\alpha + i\tau)\right)
 = \int_0^\infty e^{-\alpha t}(\cos(\tau t) - 1) \, d(P(t)+Q(t))
 = 0,
\end{equation}
So,
\begin{equation}
\int_0^\infty e^{-\alpha t}(\cos(\tau t) - 1) \, d(P(t)+Q(t)) = 0,
\end{equation}
and since
\begin{equation}
e^{-\alpha t}(\cos(\tau t) - 1) \leq 0
\end{equation}
for all $t \geq 0$, we can conclude that
\begin{equation}
\cos(\tau t) - 1 = 0 \quad \hbox{$P+Q$-a.e.},
\end{equation}
which requires $\tau = 0$ or $\displaystyle t \in \frac{2\pi}{\tau} \mathbb{Z}$. But if it's the latter, then $P+Q$, and hence $P$ and $Q$, are necessarily $\displaystyle \frac{2\pi}{\tau}$-lattice distributions, which is a contradiction. So $\tau = 0$.
\end{proof}

Define the following constants:
\begin{equation}
c_1 = \frac{p^*(\alpha)q^*(\alpha)}{\alpha \psi'(\alpha)} ,\quad
c_2 = \frac{[q^*(\alpha)]^2}{\alpha \psi'(\alpha)} ,\quad
d_1 = \frac{[p^*(\alpha)]^2}{\alpha \psi'(\alpha)} .
\end{equation}

\begin{prop}\label{mstarPole}
$m_f^*$, $m_g^*$, $n_f^*$, and $n_g^*$ have a poles at $\alpha$ of order 1 with residues $c_1$, $c_2$, $d_1$, $c_1$ respectively.
\end{prop}

Now, we will use the Wiener-Ikehara Tauberian Proposition to show $m$ converges to an exponential function if $P$ is not a lattice distribution.

\begin{thm}[Wiener-Ikehara Theorem, \cite{widder} page 233] \label{ikehara}
If $\varphi(t)$ is a non-negative, non-decreasing function for $t \geq 0$ such that the integral
\begin{equation}
f(s) = \int_0^\infty e^{-st} \varphi(t) \, dt, \quad s=\sigma+i\tau \in \mathbb{C}
\end{equation}
converges for $\sigma > 1$, and if for some constants $A\in \mathbb C$, $\alpha>0$, and some function $g(\tau)$
\begin{equation}
\lim_{\sigma \to 1^+} f(s) - \frac{A}{s-\alpha} = g(\tau)
\end{equation}
uniformly in every finite interval $-a \leq \tau \leq a$, then
\begin{equation}
\lim_{t \to \infty} \varphi(t) e^{-\alpha t} = A.
\end{equation}
\end{thm}

Note that the result is stated for $\alpha=1$ in \cite{widder}, but the general case is easily seen to follow.

\begin{thm} \label{expected-value-converges}
Suppose $P$ and $Q$ are not both $\delta$-lattice distributions.
Then,
\begin{equation}
m_f(t) \sim c_1 e^{\alpha t}, \quad
m_g(t) \sim c_2 e^{\alpha t}, \quad
n_f(t) \sim d_1 e^{\alpha t}, \quad
n_g(t) \sim c_1 e^{\alpha t}.
\end{equation}
\end{thm}

\subsection{Convergence of $X_t/\mathbb{E}[X_t]$ and $Y_t/\mathbb{E}[Y_t]$}

Let $F$ represent the joint generating function of $X_t$, $Y_t$, $X_{t+\tau}$, and $Y_{t+\tau}$ when you start with exactly one mother cell and no daughter cells, and let $G$ represent the joint generating function when you start with exactly one daughter cell and no mother cells. Then
\begin{gather}
\begin{aligned}
F(t,\tau, x_1,x_2,y_1,y_2) &= x_1x_2(1-P(t+\tau)) \\
&\phantom{{}={}} + \int_0^t F(t-y,\tau,x_1,x_2,y_1,y_2)G(t-y,\tau,x_1,x_2,y_1,y_2) \, dP(y) \\
&\phantom{{}={}} + x_1 \int_t^\tau f(t+\tau-y,x_2,y_2)g(t+\tau-y,x_2,y_2) \, dP(y)
\end{aligned} \\
\begin{aligned}
G(t,\tau, x_1,x_2,y_1,y_2) &= y_1y_2(1-Q(t+\tau)) \\
&\phantom{{}={}} + \int_0^t F(t-y,\tau,x_1,x_2,y_1,y_2)G(t-y,\tau,x_1,x_2,y_1,y_2) \, dQ(y) \\
&\phantom{{}={}} y_1 \int_t^\tau f(t+\tau-y,x_2,y_2)g(t+\tau-y,x_2,y_2) \, dQ(y),
\end{aligned}
\end{gather}
where the parameters $x_1$ and $y_1$ correspond to to $X_t$ and $Y_t$, and $x_2$ and $y_2$ correspond to $X_{t+\tau}$ and $Y_{t+\tau}$.

Define the following expectations:
\begin{equation}
\begin{aligned}
m_{2,f}(t,\tau) &:= \mathbb{E}[X_t X_{t+\tau} \, | \, X_0 = 1, \, Y_0 = 0] = F_{x_1 x_2}(t,\tau,1,1,1,1); \\
m_{2,g}(t,\tau) &:= \mathbb{E}[X_t X_{t+\tau} \, | \, X_0 = 0, \, Y_0 = 1] = G_{x_1 x_2}(t,\tau,1,1,1,1); \\
n_{2,f}(t,\tau) &:= \mathbb{E}[Y_t Y_{t+\tau} \, | \, X_0 = 1, \, Y_0 = 0] = F_{y_1 y_2}(t,\tau,1,1,1,1); \\
n_{2,g}(t,\tau) &:= \mathbb{E}[Y_t Y_{t+\tau} \, | \, X_0 = 0, \, Y_0 = 1] = G_{y_1 y_2}(t,\tau,1,1,1,1); \\
c_f(t) &:= \mathbb{E}[X_t Y_t \, | \, X_0 = 1, \, Y_0 = 0] = F_{x_1 y_1}(t,t,1,1,1,1); \\
c_g(t) &:= \mathbb{E}[X_t Y_t \, | \, X_0 = 0, \, Y_0 = 1] = G_{x_1 y_1}(t,t,1,1,1,1).
\end{aligned}
\end{equation}

\begin{prop}
The expectations $m_{2,f}$ and$m_{2,g}$ satisfy the system
\begin{gather}
\begin{aligned}
m_{2,f}(t,\tau) & = 1-P(t+\tau) \\
& \phantom{{}={}} + \int_0^t m_{2,f}(t-y,\tau) + m_{2,g}(t-y,\tau) \, dP(y) \\
& \phantom{{}={}} + \int_0^t m_f(t-y)m_g(t+\tau-y)+m_f(t+\tau-y)m_g(t-y) \, dP(y) \\
& \phantom{{}={}} + \int_t^{t+\tau} m_f(t+\tau-y) + m_g(t+\tau-y) \, dP(y), \label{m2f}
\end{aligned} \\
\begin{aligned}
m_{2,g}(t,\tau) & = \int_0^t m_{2,f}(t-y,\tau) + m_{2,g}(t-y,\tau) \, dQ(y) \\
& \phantom{{}={}} + \int_0^t m_f(t-y)m_g(t+\tau-y)+m_f(t+\tau-y)m_g(t-y) \, dQ(y). \label{m2g}
\end{aligned}
\end{gather}
Similarly, the expectations $n_{2,f}$ and $n_{2,g}$ satisfy the system
\begin{gather}
\begin{aligned}
n_{2,f}(t,\tau) & = \int_0^t n_{2,f}(t-y,\tau) + n_{2,g}(t-y,\tau) \, dP(y) \\
& \phantom{{}={}} + \int_0^t n_f(t-y)n_g(t+\tau-y)+n_f(t+\tau-y)n_g(t-y) \, dP(y), \label{n2f}
\end{aligned} \\
\begin{aligned}
n_{2,g}(t,\tau) & = 1-Q(t+\tau) \\
& \phantom{{}={}} + \int_0^t n_{2,f}(t-y,\tau) + n_{2,g}(t-y,\tau) \, dQ(y) \\
& \phantom{{}={}} + \int_0^t n_f(t-y)m_g(t+\tau-y)+n_f(t+\tau-y)m_g(t-y) \, dQ(y) \nonumber \\
& \quad + \int_t^{t+\tau} n_f(t+\tau-y) + n_g(t+\tau-y) \, dQ(y), \label{n2g}
\end{aligned}
\end{gather}
and the expectations $c_{f}$ and $c_{g}$ satisfy the system
\begin{gather}
\begin{aligned}
c_{f}(t) & = \int_0^t c_{f}(t-y) + c_{g}(t-y) \, dP(y) \nonumber \\
& \phantom{{}={}} + \int_0^t m_f(t-y)n_g(t-y)+n_f(t-y)m_g(t-y) \, dP(y), \label{cf}
\end{aligned} \\
\begin{aligned}
c_{g}(t) & = \int_0^t c_{f}(t-y) + c_{g}(t-y) \, dQ(y) \nonumber \\
& \phantom{{}={}} + \int_0^t m_f(t-y)n_g(t-y)+n_f(t-y)m_g(t-y) \, dQ(y). \label{cg}
\end{aligned}
\end{gather}
\end{prop}

Following the arguments given in Proposition~\ref{mnFinite}, we obtain the following.

\begin{prop}
The expectations $m_{2,f}(t,\tau)$, $m_{2,g}(t,\tau)$, $n_{2,f}(t,\tau)$, $n_{2,g}(t,\tau)$ are non-decreasing in both arguments and finite for $t, \tau \geq 0$.  Similarly $c_f(t)$ and $c_g(t)$ are non-decreasing and finite for $t \ge 0$.
\end{prop}

\begin{lemma}[Lemma 2 from \cite{bellmanHarris}] \label{int-limit}
If $v(t)$ satisfies the equation
\begin{equation}
v(t) = \int_0^t v(t-y) dH(y) + h(t)
\end{equation}
where $H$ is a non-decreasing function with $H(0) = 0$ and $H(\infty) = \alpha < 1$, and $h(t)$ is a bounded function such that $lim_{t\to\infty} h(t) = c$, then
\begin{equation}
\lim_{t\to\infty} v(t) = \frac{c}{1-\alpha}
\end{equation}
\end{lemma}

Now, we can prove the following convergence theorem about $\mathbb{E}[X_t X_{t+\tau}]$ and $\mathbb{E}[Y_t Y_{t+\tau}]$.

\begin{thm}\label{m2fn2fsim}
If $P$ and $Q$ is not both $\delta$-lattice distributions, then
\begin{equation}
m_{2,f}(t,\tau) \sim D_1 e^{\alpha t} e^{\alpha (t+\tau)}, \quad
n_{2,f}(t,\tau) \sim D_2 e^{\alpha t} e^{\alpha (t+\tau)},
\end{equation}
uniformly in $\tau$, and
\begin{equation}
c_{f}(t) \sim D_1 e^{2\alpha t},
\end{equation}
where
\begin{equation}
D_1 = \frac{c_1c_2 p^*(2\alpha)}{\psi(2\alpha)} , \quad
D_2 = \frac{c_1c_2 q^*(2\alpha)}{\psi(2\alpha)} .
\end{equation}
There are similar results for $m_{2,g}$, $n_{2,g}$ and $c_g$.
\end{thm}
\begin{proof}
Multiply both sides of \eqref{m2f} by $e^{-\alpha t}e^{-\alpha(t+\tau)}$, and multiply both sides of \eqref{m2g} by $e^{-\alpha t}e^{-\alpha(t+\tau)}$, and set
\begin{gather}
K_1(t,\tau) := e^{-\alpha t}e^{-\alpha(t+\tau)} m_{2,f}(t,\tau), \quad
K_2(t,\tau) := e^{-\alpha t}e^{-\alpha(t+\tau)} m_{2,g}(t,\tau), \\
d\overline{P}(y) : = \frac{e^{-2\alpha y}dP(y)}{p^*(2\alpha)}, \quad
d\overline{Q}(y) : = \frac{e^{-2\alpha y}dQ(y)}{q^*(2\alpha)}
\end{gather}
to get
\begin{equation}
\begin{aligned}
K_1(t,\tau) = p^*(2\alpha) \int_0^t K_1(t-y,\tau)+K_2(t-y) \, d\overline{P}(y) + h_1(t,\tau), \\
K_2(t,\tau) = q^*(2\alpha) \int_0^t K_1(t-y,\tau)+K_2(t-y) \, d\overline{Q}(y) + h_2(t,\tau)
\end{aligned}
\label{k12}
\end{equation}
where
\begin{equation}
\begin{aligned}
h_1(t,\tau) & := e^{-\alpha t} e^{-\alpha(t+\tau)} (1+P(t+\tau)) \\
& \phantom{{}={}} + p^*(2\alpha) \int_0^t e^{-\alpha(t-y)}m_f(t-y) e^{-\alpha(t+\tau-y)}m_g(t+\tau-y) \, d\overline{P}(y) \\
& \phantom{{}={}} + p^*(2\alpha) \int_0^t e^{-\alpha(t-y)}m_g(t-y) e^{-\alpha(t+\tau-y)}m_f(t+\tau-y) \, d\overline{P}(y) \\
& \phantom{{}={}} + p^*(2\alpha) \int_t^{t+\tau} e^{-\alpha(t-y)} e^{-\alpha(t+\tau-y)} m_f(t+\tau - y) \, d\overline{P}(y) \\
&\quad + p^*(2\alpha) \int_t^{t+\tau} e^{-\alpha(t-y)} e^{-\alpha(t+\tau-y)} m_g(t+\tau - y)  \, d\overline{P}(y).
\end{aligned}
\end{equation}
and
\begin{equation}
\begin{aligned}
h_2(t,\tau) & := q^*(2\alpha) \int_0^t e^{-\alpha(t-y)}m_f(t-y) e^{-\alpha(t+\tau-y)}m_g(t+\tau-y) \, d\overline{Q}(y) \\
& \phantom{{}={}} + q^*(2\alpha) \int_0^t e^{-\alpha(t-y)}m_g(t-y) e^{-\alpha(t+\tau-y)}m_f(t+\tau-y) \, d\overline{Q}(y).
\end{aligned}
\end{equation}
Using Theorem~\ref{expected-value-converges}, a standard $\epsilon$-$\delta$ argument shows:
\begin{equation}
\lim_{t \to \infty} h_1(t,\tau) = 2p^*(2\alpha) c_1c_2 , \quad
\lim_{t \to \infty} h_2(t,\tau) = 2q^*(2\alpha) c_1c_2 \label{m2f-limit}
\end{equation}
uniformly in $\tau$. Now, adding Equations~\eqref{k12} and setting
\begin{equation}
\begin{aligned}
K(t,\tau) &:= K_1(t,\tau) + K_2(t,\tau), \\
h(t,\tau) &:= h_1(t,\tau) + h_2(t,\tau), \\
d\overline{R}(y) &:= \frac{e^{-2\alpha y}d(P(y)+Q(y))}{p^*(2\alpha)+q^*(2\alpha)},
\end{aligned}
\end{equation}
we get
\begin{equation}
K(t,\tau) = (p^*(2\alpha) + q^*(2\alpha)) \int_0^t K(t-y,\tau) \, d\overline{R}(y) + h(t,\tau). 
\end{equation}
Since $p^*+q^*$ is a decreasing function and $p^*(\alpha) + q^*(\alpha) = 1$, then we must have $p^*(2\alpha) + q^*(2\alpha) < 1$.
Moreover, from \eqref{m2f-limit}
\begin{equation}
\lim_{t \to \infty} h(t,\tau) = 2(p^*(2\alpha) + q^*(2\alpha)) c_1c_2
\end{equation}
uniformly in $\tau$. Using Lemma~\ref{int-limit}, we can conclude we can conclude that
\begin{equation}
\lim_{t \to \infty} K(t,\tau) = \frac{2(p^*(2\alpha) + q^*(2\alpha)) c_1c_2}{\psi(2\alpha)} \label{Klimit}.
\end{equation}

Looking back at \eqref{k12}, and using~\eqref{m2f-limit} and~\eqref{Klimit}, we obtain the equation for $m_f(t,\tau)$.  The other equations follow similarly.
\end{proof}

Now, define
\begin{equation}
W_t := \frac{X_t}{c_1 e^{\alpha t}}, \quad
V_t := \frac{Y_t}{d_1 e^{\alpha t}}.
\end{equation}

\begin{prop}
If $P$ and $Q$ are not both $\delta$-lattice distributions, then
\begin{equation}
\lim_{t \to \infty} \mathbb{E}\left[\left(W_{t+\tau} - W_t \right)^2 \right] = 0, \quad
\lim_{t \to \infty} \mathbb{E}\left[\left(V_{t+\tau} - V_t \right)^2 \right] = 0
\end{equation}
uniformly in $\tau$, and
\begin{equation}
\mathbb{E}\left[\left(W_t - V_t \right)^2 \right] \to 0
\end{equation}
\end{prop}

\begin{proof}
First we show these results conditionally on either $X_0 = 1$ and $Y_0 = 0$, or $X_0 = 0$ and $Y_0 = 1$.  The case when $X_0 = 1$ and $Y_0 = 0$ follows by multiplying out $\mathbb{E}\left[\left(W_{t+\tau} - W_t \right)^2 \right] $ and applying Theorem~\ref{m2fn2fsim} to $m_f(t,t)$, $m_f(t,\tau)$, and $m_f(\tau,\tau)$.  The other cases follow similarly.

Now, suppose that $X_0$ and $Y_0$ are not specified.  Then we can create $X_0$ random variables $(\tilde X^{(n)}_t,\tilde Y^{(n)}_t)$ for $1 \le n \le X_0$, each being an independent copy of $(X_t,Y_t)$ conditioned on $X_0=1$, $Y_0 = 0$, and $Y_0$ random variables $(\hat X^{(n)}_t,\hat Y^{(n)}_t)$ for $1 \le n \le Y_0$, each being an independent copy of $(X_t,Y_t)$ conditioned on $X_0=0$, $Y_0 = 1$.  Then
\begin{equation}
\begin{aligned}
X_t &= \sum_{n=1}^{X_0} \tilde X^{(n)}_t + \sum_{n=1}^{Y_0} \hat X^{(n)}_t \\
Y_t &= \sum_{n=1}^{X_0} \tilde Y^{(n)}_t + \sum_{n=1}^{Y_0} \hat Y^{(n)}_t
\end{aligned}
\end{equation}
A short verification shows that the result still holds.
\end{proof}

Since $L^2$ is a complete metric space, we get the immediate corollary:

\begin{cor} \label{cor w v}
$W_t$ and $V_t$ converge in $L^2$ to a random variable $W$.
\end{cor}

\section{Distribution of mutant cells} \label{mutants}

In this section, we will determine a probability distribution for the mutant cell population size when the total cell population is effectively infinite. We start by making assumptions about the cell population and the expected number of mutations that occur at any given time. The culmination of this paper are Theorems~\ref{multiMutantDist} and~\ref{singleMutantDist}, where we give an explicit formula for the generating functions for the mutant cell population.

\subsection{Preliminaries and assumptions}

We start by stating our assumptions about cell growth under asymmetric cell division in which cells can mutate.

\begin{assume} Cells have the following properties:
\begin{enumerate}
\item there are exactly two types of cells: mother cells and daughter cells;
\item all cells are independent of each other, mother cells are identical to other mother cells, and daughter cells are identical to other daughter cells;
\item cell life-span for mother cells and daughter cells are strictly positive, real-valued random variables $\mathcal{T}_x$ and $\mathcal{T}_y$, respectively, with distributions $P$ and $Q$, respectively;
\item at the end of a cell's life, both mother cells and daughter cells will divide into one mother cell and one daughter cell;
\item when a cell divides, exactly one child cell can mutate with probability $\mu$;
\item all children spawned from a mutant cell will be mutants with no chance of losing the mutation.
\end{enumerate}
\end{assume}

We will be working backwards in time and let the current (time 0) cell population be $n$. In order for our model to work, we have to assume the following about our mutation rate and current cell population:

\begin{assume} \label{nLarge}
We will assume $n$ is very large and $\mu$ is sufficiently small so the product of $\mu$ and $n$ stays fixed as $\mu \to 0^+$ and $n \to \infty$, with
\begin{equation}
m := \mu n \label{def m}
\end{equation}
\end{assume}

\subsection{Constructing the generating function}
We will start by assuming time is discrete and takes values $k\delta_t$ for integers $k \geq 0$, where $\delta_t > 0$. Because we are working backwards in time, $k \delta_t$ represents time $k$ time-units ago.

Consider 4 random variables:
\begin{enumerate}
\item the number of mutant mother cells $N^M_k$ \emph{created} $k$ time-units ago;
\item the number of mutant daughter cells $N^D_k$ \emph{created} $k$ time-units ago;
\item the number of mutants $M_k$ that arise from any one mutant mother cell which was created at $k$ time-units ago;
\item the number of mutants $D_k$ that arise from any one mutant daughter cell which was created at $k$ time-units ago.
\end{enumerate}
Then the total number of mutants will be given by the formula
\begin{equation}
R = R_M + R_D,
\end{equation}
where
\begin{equation}
R_M = \sum_{k=0}^\infty \sum_{j=1}^{N^M_k} M^{(j)}_k, \quad
R_D = \sum_{k=0}^\infty \sum_{j=1}^{N^D_k} D^{(j)}_k, \label{kSumIndex}
\end{equation}
and where $M^{(j)}_k$ and $D^{(j)}_k$ denote independent copies of $M_k$ and $D_k$, respectively. 

Let $F_m(t)$ and $F_D(t)$ be functions representing the mother and daughter cell population $t$ time-units ago, respectively. So, $F_M(0) + F_D(0) = n$.

\begin{assume} \label{multiPoisson}
We will assume $N^M_k$ is a Poisson random variable such that
\begin{equation}
N^M_k \sim \operatorname{Pois}\left(\mu \left( F_M(k\delta_t)-F_M((k+1)\delta_t) \right) \right),
\end{equation}
which is assuming on average, a $\mu$-proportion of new mother cells created during that time-unit will become mutants. Similarly, we will assume $N^D_k$ is a Poisson random variable such that
\begin{equation}
N^D_k \sim \operatorname{Pois}\left(\mu \left( F_D(k\delta_t)-F_D((k+1)\delta_t) \right) \right),
\end{equation}
which is assuming on average, a $\mu$-proportion of new daughter cells created during that time-unit will become mutants.
\end{assume}

\begin{rem}
In practice, $k$ in Equations~\eqref{kSumIndex} only needs to increase to a point where the population of non-mutants $k$ time-units is much smaller than the current population. However, since we will assume we start with an effectively infinite cell population, mathematically, letting $k \to \infty$ is reasonable.
\end{rem}

\begin{assume} \label{Rindep}
In our model, we are counting the cell population that grows from a single mutant cell, and because cell growth is independent, we can assume $R_M$ and $R_D$ are independent.
\end{assume}

Using Assumption \ref{multiPoisson}, we obtain the generating function for $R_M$
\begin{equation}
\begin{aligned}
G_{R_M}(x) & = \prod_{k=0}^\infty G_{N^M_k}(G_{M_k}(x)) \\
& = \prod_{k=0}^\infty \exp\left( \mu \left( F_M(k\delta_t)-F_M((k+1)\delta_t) \right)(G_{M_k}(x)-1) \right) \\
& = \exp\left(-\mu F_M(0) + \mu \sum_{k=0}^\infty \left( F_M(k\delta_t)-F_M((k+1)\delta_t) \right) G_{M_k}(x)\right) \label{GrM}.
\end{aligned}
\end{equation}
Letting $\delta_t \to 0$ and replacing $G_{M_k}(x)$ with $f(t,x)$, we get a continuous-time model,
\begin{equation}
g_{R_M}(x) = \exp\left(-\mu F_M(0) - \mu \int_0^\infty f(t,x) \, dF_M(t) \right).
\end{equation}
Similarly,
\begin{equation}
g_{R_D}(x) = \exp\left(-\mu F_D(0) - \mu \int_0^\infty g(t,x) \, dF_D(t) \right).
\end{equation}

\begin{rem}
$f(t,x)$ and $g(t,x)$ represent the generating function for the \emph{total} cell population when you start from one mother cell or one daughter cell, respectively. This means $f(t,x) = f(t,x,x)$ and $g(t,x) = g(t,x,x)$ where the right-side of the equalities are the generating functions in \eqref{multiIntEqns}.
\end{rem}
Using Assumption \ref{Rindep}, we get
\begin{equation}
\begin{aligned}
g_R(x) & = g_{R_M}(x)g_{R_D}(x) \\
& = \exp\left(-\mu n - \mu \int_0^\infty f(t,x) \, dF_M(t) -\mu \int_0^\infty g(t,x) \, dF_D(t) \right) \label{grMulti}.
\end{aligned}
\end{equation}

Since we are working backwards in time, we can interpret our current time as infinite relative to when we started from a single cell.  Suppose $P$ and $Q$ are not both $\delta$-lattice distributions. By Theorem~\ref{main-asymp}, we know that asymptotically, the proportion of mother cells to the total cell population will be $p^*(\alpha)$, while the proportion of daughter cells to the total cell population will be $q^*(\alpha) = 1 - p^*(\alpha)$.
This leads us to the following assumption about the cell population at backwards time $t$.

\begin{assume} The mother cell population at backwards time $t$ is
\begin{equation}
F_M(t) = n p^*(\alpha) e^{-\alpha t},
\end{equation}
and the daughter cell population at backwards time $t$ is
\begin{equation}
F_d(t) = n q^*(\alpha) e^{-\alpha t}
\end{equation}
with high probability.
\end{assume}

Hence
\begin{equation}
g_R(x) = \exp \left(-m + m \alpha \int_0^\infty ( p^*(\alpha) f(t,x) + q^*(\alpha)g(t,x)) e^{-\alpha t} \, dt \right),
\end{equation}
where we recall Equation~\eqref{def m}.

The above construction culminates to the following theorem:

\begin{thm} \label{multiMutantDist}
Suppose distributions $P$ and $Q$ are not both $\delta$-lattice distributions for any $\delta>0$.  Let $R$ be a random variable representing the mutant cell population in an effectively infinite cell population $n$ with effectively zero mutation rate $\mu$, and set $m := \mu n$. Let $f(t,x)$ and $g(t,x)$ be the generating functions for the total cell population at time $t$ when you start from a single mother cell and daughter cell, respectively, with life-span distributions $P$ and $Q$ for mother cells and daughter cells, respectively, which must satisfy
\begin{equation}
\begin{aligned}
f(t,x) &= x(1-P(t)) + \int_0^t f(t-\tau,x)g(t-\tau,x) \, dP(\tau) \\
g(t,x) &= x(1-Q(t)) + \int_0^t f(t-\tau,x)g(t-\tau,x) \, dQ(\tau) .
\end{aligned}
\end{equation}
Let $\alpha$ be the root of
\begin{equation}
1- \int_0^\infty e^{-st} \, dP(t) - \int_0^\infty e^{-st} \, dQ(t) = 0
\end{equation}
and
\begin{equation}
\gamma := \int_0^\infty e^{-\alpha t} \, dP(t).
\end{equation}
Then the generating function for $R$ is
\begin{equation}
g_R(x) = \exp \left(-m + m \alpha \int_0^\infty ( \gamma f(t,x) + (1-\gamma)g(t,x)) e^{-\alpha t} \, dt \right). \label{gr2multi}
\end{equation}
\end{thm}

It is also worth stating the special case when cell division is symmetric, when we don't need to distinguish between mother and daughter cells.

\begin{thm} \label{singleMutantDist}
Suppose the distribution $P$ is not a $\delta$-lattice distribution for any $\delta>0$.  Let $R$ be a random variable representing the mutant cell population in an effectively infinite cell population $n$ with effectively zero mutation rate $\mu$, and set $m := \mu n$. Let $f(t,x)$ total cell population at time $t$ when you start from a cell with life-span distribution $P$, which must satisfy
\begin{equation}
f(t,x) = x(1-P(t)) + \int_0^t f(t-\tau,x)^2 \, dP(\tau)
\end{equation}
Let $\alpha$ be the root of
\begin{equation}
1- 2\int_0^\infty e^{-st} \, dP(t) = 0.
\end{equation}
Then the generating function for $R$ is
\begin{equation}
g_R(x) = \exp \left(-m + m \alpha \int_0^\infty f(t,x) e^{-\alpha t} \, dt \right). \label{gr2single}
\end{equation}
\end{thm}

\section{Examples} \label{MutantExamples}

\subsection{Life-span is multi-phase}

We say that $P$ has a ``multi-phase" distribution with parameter $n$ if is a gamma distribution with parameters $\beta = 1$ and $\alpha=n$. Biologically, this means that a cell is more likely to divide around the $n^{\text{th}}$ ``stage" of its life. If $n=1$, then a cell is most likely to divide closer to birth, and $P$ would be an exponential distribution. However, when $n=3$, for example, it would mean the cell takes time to mature before it can divide, which is more realistic biologically.

Suppose $P$ and $Q$ are multi-phase distributions with parameter $n$ and $m$, respectively. Then
\begin{equation}
dP(t) = \frac{t^{k-1}e^{-t}}{(k-1)!}\, dt, \quad
dQ(t) = \frac{t^{n-1}e^{-t}}{(n-1)!}\, dt
\end{equation}
and the exponential growth rate $\alpha$ for both $X_t$ and $Y_t$ is the root to
\begin{equation}
{(1+s)}^{k+n} - {(1+s)}^k - {(1+s)}^n = 0.
\end{equation}

The multi-phase distribution is of great importance in the asymmetric cell division case, and is a major inspiration for this paper, where we were interested in developing a generating function when the daughter cells take time to mature before they can divide. A good model for this would be to assume mother cells have multi-phase distribution with parameter $k=1$ or $n=2$, while daughter cells have multi-phase distribution $k=2$ or $n=3$.

\begin{thm} \label{Multiphase}
If $P$ and $Q$ are multi-phase distributions with parameters $k$ and $n$, respectively, then the generating functions $f=f(t,x,y)$ and $g=g(t,x,y)$ satisfy
\begin{equation}
\left(1+ \frac{\partial}{\partial t} \right)^k f = fg, \quad
\left(1+ \frac{\partial}{\partial t} \right)^n g = fg, \quad
f(0,x,y) = x, \quad
g(0,x,y) = y
\end{equation}
\end{thm}

\subsection{Symmetric cell division} \label{singleMutantExamples}

Suppose $P$ is an exponential distribution with parameter $1$, that is, it is a multiphase distributions with $n=1$.  Then the differential equation for the generating function is
\begin{equation}
\frac{\partial f}{\partial t} = f^2 - f, \quad
f(0,x) = x
\end{equation}
which can be solved to obtain
\begin{equation}
f(t,x) = \frac{xe^{-t}}{1-x+xe^{-t}}.
\end{equation}
Moreover $\alpha = 1$.  So, the generating function becomes
\begin{equation}
\begin{aligned}
g_R(x) &= \left( 1-x \right)^{\frac{m (1-x)}{x}} \\
&= e^{-m}+\frac{1}{2} e^{-m} m x+\frac{1}{24} e^{-m} m (3 m+4) x^2+\frac{1}{48} e^{-m} m (m+2)^2 x^3+ \dots .
\end{aligned}
\end{equation}
which is the Lea-Coulson distribution (see \cite{leaCoulson}).

\subsection{Asymmetric cell division} \label{multiMutantExamples}

Suppose $P$ and $Q$ are multi-phase distributions with parameters $1$ and $2$, respectively. Then we have
\begin{equation}
\alpha = \gamma = \frac{\sqrt{5}-1}{2}.
\end{equation}
This models a situation where mother cells have shorter expected life-spans than daughter cells.

Using the series solution to $f(t,x)$ and $g(t,x)$ discussed in Section~\ref{multiSeriesSoln}, and the Ma-Sandri-Sarkar algorithm (see \cite{mss}), we can compute the coefficients of $g_R(x)$ relatively quickly.  These were computed using the software package {\tt sagemath} (see \cite{sagemath}), using code shown in Section~\ref{sage-code}.
\Tiny
\def\nl{\\&\phantom{{}={}}}
\begin{align*}
\Pr(R = 0 ) &= e^{-m} \\
\Pr(R = 1 ) &= 0.472135955 \, m \, e^{-m} \\
\Pr(R = 2 ) &= (0.111456180002 \, m^{2} + 0.172209268743 \, m) e^{-m} \\
\Pr(R = 3 ) &= (0.0175408233286 \, m^{3} + 0.0813061875578 \, m^{2} + 0.0923265707329 \, m) e^{-m} \\
\Pr(R = 4 ) &= (0.00207041334343 \, m^{4} + 0.019193787255 \, m^{3} + 0.0584187097653 \, m^{2} + 0.054750171345 \, m) e^{-m}\\
\Pr(R = 5 ) &= (0.000195503316229 \, m^{5} + 0.00302069235857 \, m^{4} + 0.0172912064384 \, m^{3} + 0.0417490156659 \, m^{2} + 0.0359263840576 \, m) e^{-m} \\
\Pr(R = 6 ) &= (1.53840241522 \times 10^{-5} \, m^{6} + 0.000356544367868 \, m^{5} + 0.00327215809956 \, m^{4} + 0.0144601403668 \, m^{3} + 0.0306527224493 \, m^{2} \nl + 0.0253788856813 \, m) e^{-m} \\
\Pr(R = 7 ) &= (1.03762156212 \times 10^{-6} \, m^{7} + 3.36674831246 \times 10^{-5} \, m^{6} + 0.000451249775082 \, m^{5} + 0.00313432945934 \, m^{4} \nl + 0.0118370547378 \, m^{3} + 0.0232240363224 \, m^{2} + 0.0188813484058 \, m) e^{-m} \\
\Pr(R = 8 ) &= (6.12373058949 \times 10^{-8} \, m^{8} + 2.64927154958 \times 10^{-6} \, m^{7} + 4.87502731868 \times 10^{-5} \, m^{6} + 0.000487114247624 \, m^{5} \nl + 0.00283908715808 \, m^{4} + 0.00968207809151 \, m^{3} + 0.018100793276 \, m^{2} + 0.0145941835643 \, m) e^{-m} \\
\Pr(R = 9 ) &= (3.21248154448 \times 10^{-9} \, m^{9} + 1.78688050445 \times 10^{-7} \, m^{8} + 4.3192805187 \times 10^{-6} \, m^{7} + 5.85526505241 \times 10^{-5} \, m^{6} \nl + 0.000484413747255 \, m^{5} + 0.00250654733417 \, m^{4} + 0.0079759775165 \, m^{3} + 0.0144521031634 \, m^{2} + 0.0116173054839 \, m) e^{-m} \\
\Pr(R = 10 ) &= (1.51672804192 \times 10^{-10} \, m^{10} + 1.0545631668 \times 10^{-8} \, m^{9} + 3.23914598676 \times 10^{-7} \, m^{8} + 5.71296309855 \times 10^{-6} \, m^{7} \nl + 6.34669453811 \times 10^{-5} \, m^{6} + 0.000460389737452 \, m^{5} + 0.00218929525915 \, m^{4} + 0.00663573155793 \, m^{3} + 0.0117763023231 \, m^{2} \nl + 0.00946644120115 \, m) e^{-m}
\end{align*}

\normalsize

\subsection{Sage code to compute probabilities for multiphase distributions}
\label{sage-code}

\begin{verbatim}
t = var('t'); m = var('m'); y = var('y'); a = var('a')

# The next line seems to be required to make the definite integrals work.
assume(t>0)

# number of coefficients to compute.
degree = 10

# life-span distribution density of P (multi-phase with parameter k).
k = 1
p(t) = t^(k-1)*exp(-t)/factorial(k-1)
P(t) = integrate(p(y),y,0,t)

# life-span distribution density of Q (multi-phase with parameter n).
n = 2
q(t) = t^(n-1)*exp(-t)/factorial(n-1)
Q(t) = integrate(q(y),y,0,t)

# find alpha and gamma.
eqn = (1+x)^(k+n)-(1+x)^k-(1+x)^n
assume(x,'real')
alpha = max(a.rhs() for a in eqn.solve(x))
# use float for faster calculations.
alpha = float(alpha)
gamma = integrate(exp(-alpha*t)*p(t),t,0,infinity)

# store coefficients c[i] and b[i] of the generating functions f and g.
c = [0,1-P(t)]
b = [0,1-Q(t)]
for i in xrange(2,degree+1):
    tempsum(t) = sum(expand(c[j]*b[i-j]) for j in xrange(1,i))
    # indefinite integration is faster than definite integration.
    temp1(y)=integrate(expand(tempsum(y)*expand(p(t-y))),y)
    temp2(y)=integrate(expand(tempsum(y)*expand(q(t-y))),y)
    c.append(temp1(t)-temp1(0))
    b.append(temp2(t)-temp2(0))

# store coefficients of H = integral(pgf*alpha*e^(-alpha t),t,0,infinity).
h = [0]
for i in xrange(1,degree+1):
    h.append(integrate((gamma*c[i]+
                       (1-gamma)*b[i])*
                       alpha*exp(-alpha*t),t,0,infinity))

# compute coefficients of e^H using MSS algorithm.
prob = [1]
for r in xrange(1,degree+1):
    prob.append(expand(sum(m*s/r*h[s]*prob[r-s] for s in xrange (1,r+1))))

for r in xrange(0,degree+1):
    print "\\Pr(R =",r,") &=",latex(e^(-m)*prob[r]),"\\\\"
\end{verbatim}

\subsection{Mathematica code to compute probabilities for multiphase distributions}

\begin{verbatim}
degree = 10;
k = 1;
p[t_] = t^(k - 1) Exp[-t]/(k - 1)!;
P[t_] = Integrate[p[y], {y, 0, t}];
n = 2;
q[t_] = t^(n - 1) Exp[-t]/(n - 1)!;
Q[t_] = Integrate[q[y], {y, 0, t}];
alpha = N[Max[x /. Solve[(1 + x)^(k + n) - (1 + x)^k - (1 + x)^n == 0, x]]];
gamma = Integrate[Exp[-alpha t] p[t], {t, 0, Infinity}];
c[1] = 1 - P[t];
b[1] = 1 - Q[t];
Do[tempsum[t_] = Sum[Expand[c[j] b[i - j]], {j, 1, i - 1}];
   c[i] = Integrate[Expand[tempsum[y] Expand[p[t - y]]], {y, 0, t}]; 
   b[i] = Integrate[Expand[tempsum[y] Expand[q[t - y]]], {y, 0, t}],
   {i, 2, degree}];
Do[h[i] = Integrate[(gamma c[i] + (1 - gamma) b[i]) alpha Exp[-alpha t],
                    {t, 0, Infinity}],
   {i, 1, degree}];
prob[0] = 1;
Do[prob[r] = Expand[Sum[m s/r h[s] prob[r - s], {s, 1, r}]], {r, 1, degree}];
Table[{r, prob[r] Exp[-m]}, {r, 0, degree}] // TableForm
\end{verbatim}

\section*{Acknowledgments}

We would like to express our sincere thanks to Professor George Smith for bringing this problem to our attention, and many valuable discussions.

Some of this research was funded by NSF DMS Grant 0928053, PRISM: Mathematics in Life Sciences.


\bibliographystyle{plain}
\nocite{*}
\bibliography{references}

\end{document}